\documentclass{article}
\usepackage[a4paper, total={6in, 8in}]{geometry}
\usepackage[utf8]{inputenc}
\usepackage{graphicx} 
\usepackage{tikz}
\usepackage{dsfont}
\usetikzlibrary{calc}
\usepackage{graphicx}
\usepackage{lipsum}
\usepackage{float}
\usepackage[rightcaption]{sidecap}
\graphicspath{ {./images/} }
\usepackage[utf8]{inputenc}
\usepackage{amsmath,amssymb,amsthm}
\usepackage{lmodern,stmaryrd}
\usepackage{wasysym}
\usepackage{sgame}
\usepackage[T1]{fontenc}
\usepackage[export]{adjustbox}
\newtheorem{theorem}{Theorem}
\newtheorem{definition}{Definition}
\newtheorem{lemma}{Lemma}

\newtheorem{assumption}{Assumption}

\counterwithin*{equation}{section}
\counterwithin*{equation}{subsection}
\usepackage[pdftex]{hyperref}
\usepackage{chicago}
\usepackage{enumerate,tikz,arydshln,ragged2e,wasysym}
\DeclareMathOperator*{\E}{\mathbb{E}}
\usepackage{apacite}

\usepackage{bbm}

\begin{document}

\thispagestyle{plain}
\begin{center}
    \Large
    \textbf{Rationalizability and Monotonocity in Games with Incomplete Information \footnote{I would like to acknowledge Andrés Perea for his valuable contributions as my supervisor. His guidance, availability for discussions, and willingness to provide constructive criticism have been greatly appreciated.} }
        
   \vspace{0.4cm}
   \large

   \vspace{0.4cm}
    \text{Joep van Sloun} \footnote{Department of Quantitative Economics, School of Business and Economics, Maastricht University, 6200 MD
Maastricht, THE NETHERLANDS, Email: j.vansloun@maastrichtuniversity.nl
}
       \large

    \vspace{0.4cm}
    \text{\today}
    \vspace{0.9cm}

    \end{center}

\begin{abstract}
 This paper examines games with strategic complements or substitutes and incomplete information, where players are uncertain about the opponents' parameters. We assume that the players' beliefs about the opponent's parameters are selected from some given set of beliefs. One extreme is the case where these sets only contain a single belief, representing a scenario where the players' actual beliefs about the parameters are commonly known among the players. Another extreme is the situation where these sets contain all possible beliefs, representing a scenario where the players have no information about the opponents' beliefs about parameters. But we also allow for intermediate cases, where these sets contain some, but not all, possible beliefs about the parameters. We introduce an assumption of weakly increasing differences that takes both the choice belief and parameter belief of a player into account. Under this assumption, we demonstrate that greater choice-parameter beliefs leads to greater optimal choices. Moreover, we show that the greatest and least point rationalizable choice of a player is increasing in their parameter, and these can be determined through an iterative procedure. In each round of the iterative procedure, the lowest surviving choice is optimal for the lowest choice-parameter belief, while the greatest surviving choice is optimal for the highest choice-parameter belief.

\end{abstract}
\bigskip

            \vspace{0in}

\newpage
\section{Introduction}
We focus on games with strategic complements or substitutes, a concept introduced by \citeA{bulow1985multimarket}. In games with strategic complements, a player's utility increases when an opponent increases their choice, while in games with strategic substitutes, it is the opposite. Many real-world games exhibit strategic complementarities and substitutes, such as oligopolies, games with network externalities, and bank runs. Extensive research has been done on these games under complete information, see for example \citeA{milgrom1990rationalizability}, \citeA{vives1990nash}, and \citeA{topkis1979equilibrium} showing the existence of largest and smallest equilibria in pure choices that bound the set of rationalizable choices.
\newline
\newline
Similar results have been obtained in the literature on games with incomplete information. \citeA{vives1990nash} showed that his results for the complete information case could be extended to the case for incomplete information. \citeA{athey2001single} studied games with a single crossing condition, where a player's optimal choice increases when each opponent uses an increasing strategy, i.e., choosing a greater choice when they have a higher parameter value compared to a lower parameter value. This result has been further extended and generalized by \citeA{mcadams2003isotone} and \citeA{reny2011existence}. \citeA{van2007monotone} considered games where a player's belief about their opponents' parameters changes with their own parameter, specifically, if a player's parameter increases, they believe that greater parameter values of their opponents are more likely. This condition is weaker than the affiliation condition used in \citeA{athey2002monotone}. 
\newline
\newline
The existing literature often assumes that the distribution of the players' parameter values is commonly known by the players, which represents an extreme case of full information about the distribution of parameter values. On the other hand, another extreme case is when players believe they have no knowledge about the probability distribution of opponents' parameter values. In that case, a player could form any belief about the distribution of the parameter values of his opponents. In this paper, we consider a more nuanced approach where players form beliefs from a given set of parameter beliefs, which can range from these extreme cases to intermediate cases.  
\newline
\newline
Furthermore, similar to the related literature, we assume that players also form choice beliefs, which specify the choices they believe their opponents will make for each possible parameter value. By combining choice beliefs with parameter beliefs, players can compute the probability of each choice combination of their opponents being played. We then introduce a notion of stochastic dominance, where a choice-parameter belief pair stochastically dominates another pair if the probability of greater choices is higher under the former.
\newline
\newline
The solution concept we use is point rationalizability, introduced by \citeA{bernheim1984rationalizable}, which has been generalized for games with incomplete information. For example, the solution concept of interim rationalizability has been proposed by \citeA{ely2004hierarchies} and \citeA{dekel2007interim}. While interim rationalizability fixes the belief hierarchies on utilities, we do not, making our solution concept more akin to $\Delta$-rationalizability proposed by \citeA{battigalli2003rationalizability} and \citeA{battigalli2003rationalization}. Unlike $\Delta$-rationalizability, we impose no constraints on the first-order choice beliefs of players. Our approach is also similar to that of \citeA{bach2021incomplete}, but with the distinction that the sets of choices and parameters in our setting are not finite. Moreover, we restrict to point beliefs, assigning probability 1 to a single choice for every parameter value of the opponent.
\newline
\newline
By focusing on point beliefs, we simplify the comparison with pure Nash equilibria and streamline the analysis. Rationalizability relaxes the assumption of correct beliefs, which requires each player to believe that his opponents are correct about his own beliefs and that his opponents share his own beliefs about other players. When each player can have only one parameter belief, he automatically believes that his opponents are correct about his parameter belief. However, when players can have multiple parameter beliefs, it becomes possible for a player to believe that his opponents are incorrect about his parameter belief, which relaxes the correct beliefs assumption for both choice and parameter beliefs.
\newline
\newline
In this paper, we provide a characterization of the set of point rationalizable choices for games with strategic complements or substitutes. We make the assumption that the expected utility function of each player exhibits weakly increasing differences. This means that an increase in the parameter value or the probability of high choices will result in an increase in the expected utility difference between two choices of a player. Specifically, we show that in our iterative procedure to characterize the point rationalizable choices, a player's highest and lowest combination of choice and parameter belief determines the highest and lowest choices that the player can rationally make. Furthermore, we demonstrate that as a player's parameter increases, both their lowest and greatest point rationalizable choice also increase. Moreover, we aim to shed light on how players' decision-making processes changes when they are able to compare different choice and parameter belief pairs. We demonstrate that in games with strategic complements (substitutes), where the utility function of a player has a third-order partial derivative of zero with respect to his choice and any two choices of his opponents, an increase in the choice-parameter belief pair with respect to stochastic dominance leads to an increase (decrease) in the optimal choice of the player.
\newline
\newline
The outline of the paper is as follows: Section 2 provides the necessary definitions and introduces point rationalizability under incomplete information. Section 3 shows the main assumptions and results. In Section 4, we provide two examples to illustrate our findings. The first example is a Bertrand model with differentiated goods, where players face complete uncertainty about costs of their opponent. The second example is a Cournot model where players face uncertainty about the price elasticity of demand. Section 5 provides some concluding remarks. All proofs and lemmas used in the proofs are collected in the appendix.

\section{Point Rationalizability in Games with Incomplete Information}
\subsection{Game with Incomplete Information}
We start by defining a game with incomplete information.
\begin{definition}
    A game with incomplete information $(C_i,u_i,\Theta_i)_{i \in N}$ is specified by the following: \newline
    1. A set of players $N=\{1,...,n\}$. \newline
    2. A set of choices $C_i$ for each player $i \in N$. \newline
    3. A set $\Theta_i$ of possible parameter values for each player $i \in N$. \newline
    4. A utility function $u_i: \Theta_i \times C_i \times C_{-i} \xrightarrow{}\mathbb{R}$ for each player $i \in N$, where $ C_{-i}= \underset{j\neq i}{\times} C_j $
\end{definition}
 We assume that $\Theta_i=[ \underline{\theta}_i,\Bar{\theta}_i]$ for every player $i$ and that $c_i \in C_i=[\underline{c}_i,\Bar{c}_i] \ \forall i \in N$. Finally, let $\Theta_{-i}=\underset{j\neq i}{\times} \Theta_{j}$. 
 \newline
 \newline
 Each player $i$ knows exactly what his own parameter value is, but is uncertain about the parameter values of his opponents, which makes him uncertain about his opponents' utility functions as well. Hence, a player forms a probabilistic parameter belief about the parameter values of the opponents. This is denoted by $f_i= (f_{ij})_{j \neq i }$, where $f_{ij}$ is a probability density function on $\Theta_j$. The set of all possible parameter beliefs of a player $i$ is denoted by $M_i$. A choice belief $\beta_i$ specifies a choice for each opponent $j$, for each possible parameter combination $\theta_j \in \Theta_j$. Hence, a choice belief is a point belief. We denote this as $\beta_i=(\beta_{ij}(\theta_j)_{\theta_j \in \Theta_j})_{j \neq i}$, where $\beta_{ij}(\theta_j) \in C_j$ for every $j \neq i$ and $\theta_j \in \Theta_j$.
\newline
\newline
A player $i$ can combine these two beliefs to compute the probability of a choice combination of the opponents being played. 

\begin{definition}
Consider player $i$ with a parameter belief $f_i$ and a choice belief $\beta_i$. The composite belief $(\beta_i \circ f_i)$ which combines these two beliefs is given by
$$(\beta_i \circ f_i)(c_{-i})=  \prod_{j\neq i} (\beta_{ij} \circ f_{ij})(c_j) $$
where
$$(\beta_{ij} \circ f_{ij})(c_j)= \int_{\theta_j:\beta_{ij}(\theta_j)=c_{j}}  f_{ij}(\theta_j) \, d \theta_j.$$
\end{definition}

A composite belief is useful because even though a player forms a choice belief and a parameter belief, he is ultimately interested in the probability of each choice combination of his opponents. A player can use this composite belief to calculate his expected utility.

\begin{definition}
 Consider player $i$ with a parameter belief $f_i$ and a choice belief $\beta_i$. We can denote the expected utility of player $i$ as 

 $$U_i(\theta_i,c_i,\beta_i,f_i)= \int_{c_{-i} \in C_{-i}} (\beta_i \circ f_i)(c_{-i}) \cdot  u_i(\theta_i, c_i,c_{-i}) \, d c_{-i} $$

  Equivalently, the expected utility function of player $i$ can also be directly computed given his choice belief and parameter belief.
 
  $$U_i(\theta_i,c_i,\beta_i,f_i)=  \int_{\theta_{-i} \in \underset{j\neq i}{\times} \Theta_j }  \prod_{j\neq i} f_{ij}(\theta_j) \cdot u_i(\theta_i, c_i,(\beta_{ij}(\theta_j))_{j \neq i}) \, d \theta_{-i}.$$

  \end{definition}
We will now define that a choice is optimal for a player if this choice maximizes his expected utility.
\begin{definition}
Consider player $i \in N$. The choice $c_i$ is optimal for $(\theta_i,\beta_i,f_i)$ if 
$$U_i(\theta_i,c_i,\beta_i,f_i) \geq U_i(\theta_i,c'_i,\beta_i,f_i) \ \forall \ c'_i \in C_i $$

If $c_i \in A_i \subseteq C_i$ and the above inequality holds for every $c'_i\in A_i$, then $c_i$ is optimal in $A_i$ for $(\theta_i,\beta_i,f_i)$.
\end{definition}

\subsection{Point Rationalizability}
We can now start to define point rationalizability \cite{bernheim1984rationalizable}, adapted for games with incomplete information. In each round of the iterative procedure, only the choices that are optimal for a (surviving) choice belief and a parameter belief survive the round. 

\begin{definition}
Let $P_i^0(\theta_i)= C_i \ \forall \theta_i \in \Theta_i \ \text{and} \ \forall i \in N$
Let $k \geq 1$ and suppose that $P_i^{k-1}(\theta_i)$ has been defined for all players $i \in N$ and all $\theta_i\in \Theta_i$. Then for all players $i \in N$ and all $\theta_i \in \Theta_i$,
$$P_i^k(\theta_i)= \{ c_i \in P_i^{k-1}(\theta_i) | c_i \ \text{is optimal in } P_i^{k-1}(\theta_i) \ \text{for the parameter} \ \theta_i \ \text{and a choice belief} \ \beta_i $$ $$ \ \text{and parameter belief} \ f_i \in M_i \ \text{where} \ \beta_{ij}(\theta_{j}) \in P_{j}^{k-1}(\theta_{j}) \ \forall \ \theta_j \in \theta_j \ \text{and} \ j \neq i \}.$$
\end{definition}
 The set of point rationalizable choices of an player $i$ with parameter $\theta_i$ can then be denoted as $P_i(\theta_i)= \underset{k \geq 1}{\bigcap} P_i^k(\theta_i)$.

\subsection{Comparing Parameters and Beliefs}

The next definition provides a way to compare choice combinations of opponents of an agent.

\begin{definition}
We define $c_{-i}' \geq c_{-i}$ if 

$$c_j' \geq c_j \ \forall j \neq i.$$

\end{definition}

For the next definition, a choice belief is considered larger than another choice belief if, for the former, a player believes that the opponents will choose a greater choice for each possible parameter.
\begin{definition}
We define $\beta_i' \geq \beta_i$ if 

$$\beta_{ij}'(\theta_j) \geq \beta_{ij}(\theta_j) \ \forall \theta_j \in \Theta_j \ \text{and} \ \forall \ j \neq i.$$

\end{definition}
The next definition provides a notion of comparing parameter beliefs. A parameter belief $f_i'$ first-order stochastically dominates another parameter belief $f_i$ if the probability of high parameter values is greater under $f_i'$ \cite{hadar1969rules}.  
\begin{definition}
We define $ f_{ij}' \geq f_{ij}$ if 

 $$\int_{\theta_j'}^{\Bar{\theta}_j}f_{ij}'(\theta_j)  \, d \theta_{j} \geq \int_{\theta_j'}^{\Bar{\theta}_j}f_{ij}(\theta_j)  \, d \theta_{j}  \ \forall \theta_{j}' \in \Theta_{j}   $$

 and $f_i'\geq f_i$ if  $ f_{ij}' \geq f_{ij} \ \forall j \neq i$.

\end{definition}
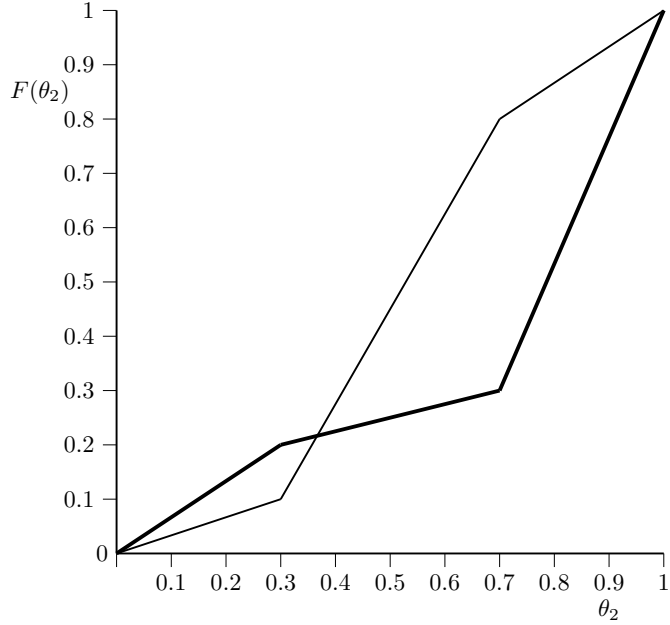
\begin{figure}[t]\label{figure: henk}
    \centering
     \scalebox{0.9} {
    \begin{tikzpicture}
[scale=0.8]
   \draw[thick] (0,0) -- (10,0);
\draw[thick] (0,0) -- (0,10);
\draw (0,0) node [anchor=east] {$0$};
\draw (-0.7,8.5) node [anchor=east] {$F(\theta_2)$};

\draw (9,-0.7) node [anchor=north] {$\theta_2$};

\draw[ultra thick] (0,0) -- (3,2);
\draw[ultra thick] (3,2) -- (7,3);
\draw[ultra thick] (7,3) -- (10,10);
\draw[thick] (0,0) -- (3,1);
\draw[thick] (3,1) -- (7,8);
\draw[thick] (7,8) -- (10,10);

  \draw[black] (0,0) -- (0, -0.25);
            \draw[black] (1,0) -- (1, -0.25);
            \draw (1,-0.25) node [anchor=north] {$0.1$};
             \draw[black] (2,0) -- (2, -0.25);
             \draw (2,-0.25) node [anchor=north] {$0.2$};
            \draw[black] (3,0) -- (3, -0.25);
            \draw (3,-0.25) node [anchor=north] {$0.3$};
             \draw[black] (4,0) -- (4, -0.25);
             \draw (4,-0.25) node [anchor=north] {$0.4$};
\draw[black] (5,0) -- (5, -0.25);
\draw (5,-0.25) node [anchor=north] {$0.5$};
             \draw[black] (6,0) -- (6, -0.25);
             \draw (6,-0.25) node [anchor=north] {$0.6$};
            \draw[black] (7,0) -- (7, -0.25);
            \draw (7,-0.25) node [anchor=north] {$0.7$};
             \draw[black] (8,0) -- (8, -0.25);
             \draw (8,-0.25) node [anchor=north] {$0.8$};
             \draw[black] (9,0) -- (9, -0.25);
             \draw (9,-0.25) node [anchor=north] {$0.9$};
             \draw[black] (10,0) -- (10, -0.25);
             \draw (10,-0.25) node [anchor=north] {$1$};

            \draw[black] (0,1) -- (-0.25,1);
            \draw (-0.25,1) node [anchor=east] {$0.1$};
             \draw[black] (0,2) -- (-0.25,2);
             \draw (-0.25,2) node [anchor=east] {$0.2$};
            \draw[black] (0,3) -- (-0.25,3);
            \draw (-0.25,3) node [anchor=east] {$0.3$};
             \draw[black] (0,4) -- (-0.25,4);
             \draw (-0.25,4) node [anchor=east] {$0.4$};
\draw[black] (0,5) -- (-0.25,5);
\draw (-0.25,5) node [anchor=east] {$0.5$};
             \draw[black] (0,6) -- (-0.25,6);
             \draw (-0.25,6) node [anchor=east] {$0.6$};
            \draw[black] (0,7) -- (-0.25,7);
            \draw (-0.25,7) node[anchor=east] {$0.7$};
             \draw[black] (0,8) -- (-0.25,8);
             \draw (-0.25,8) node [anchor=east] {$0.8$};
             \draw[black] (0,9) -- (-0.25,9);
             \draw (-0.25,9) node [anchor=east] {$0.9$};
             \draw[black] (0,10) -- (-0.25,10);
             \draw (-0.25,10) node [anchor=east] {$1$};

\end{tikzpicture}
}
   \caption{Parameter beliefs of player 1}
    \label{fig:my_label}
\end{figure}

Similarly, the next definition combines the choice belief and parameter belief. We say that  $(\beta_i' \circ f_i')$ stochastically dominates $ (\beta_i \circ f_i)$ if the probability of high choices is greater under the former belief.
\begin{definition}
We define $(\beta_{ij}' \circ f_{ij}') \geq (\beta_{ij} \circ f_{ij})$ if

$$\int_{c_j'}^{\Bar{c}_{j}} (\beta_{ij}' \circ f_{ij}')(c_{j}) \, d c_{j} \geq \int_{c_{j}'}^{\Bar{c}_{j}} (\beta_{ij} \circ f_{ij})(c_{j})  \, d c_{j}  \ \forall c_{j}'  \in C_{j} $$
and $(\beta_{i}' \circ f_{i}') \geq (\beta_{i} \circ f_{i})$ if $(\beta_{ij}' \circ f_{ij}') \geq (\beta_{ij} \circ f_{ij}) \ \forall j\neq i$.
\end{definition}
Definition 8 and 9 correspond to first-order stochastic dominance. The next definition combines our previous definitions of a composite belief and parameter value. 

\begin{definition}
We define $(\theta_i',\beta_i', f_i') \geq (\theta_i,\beta_i, f_i)$ if $\theta_i' \geq \theta_i$ and $(\beta_i' \circ f_i') \geq (\beta_i \circ f_i)$.  
\end{definition}

Consider the following simple example with 2 players, $\Theta_1= \Theta_2= [0,1]$. Consider the following parameter beliefs of player 1 

\begin{figure}[t]\label{figure: henk}
    \centering
     \scalebox{0.9} {
    \begin{tikzpicture}
[scale=0.8]
   \draw[thick] (0,0) -- (10,0);
\draw[thick] (0,0) -- (0,10);
\draw (0,0) node [anchor=east] {$0$};
\draw (-0.7,8.5) node [anchor=east] {$\beta_{12}(\theta_2)$};

\draw (9,-0.7) node [anchor=north] {$\theta_2$};

\draw[ultra thick] (0,5) -- (3,5);
\draw[ultra thick] (3,3) -- (7,3);
\draw[ultra thick] (7,8) -- (10,8);
\draw[thick] (0,8) -- (3,8);
\draw[thick] (3,2) -- (7,2);
\draw[thick] (7,5) -- (10,5);

  \draw[black] (0,0) -- (0, -0.25);
            \draw[black] (1,0) -- (1, -0.25);
            \draw (1,-0.25) node [anchor=north] {$0.1$};
             \draw[black] (2,0) -- (2, -0.25);
             \draw (2,-0.25) node [anchor=north] {$0.2$};
            \draw[black] (3,0) -- (3, -0.25);
            \draw (3,-0.25) node [anchor=north] {$0.3$};
             \draw[black] (4,0) -- (4, -0.25);
             \draw (4,-0.25) node [anchor=north] {$0.4$};
\draw[black] (5,0) -- (5, -0.25);
\draw (5,-0.25) node [anchor=north] {$0.5$};
             \draw[black] (6,0) -- (6, -0.25);
             \draw (6,-0.25) node [anchor=north] {$0.6$};
            \draw[black] (7,0) -- (7, -0.25);
            \draw (7,-0.25) node [anchor=north] {$0.7$};
             \draw[black] (8,0) -- (8, -0.25);
             \draw (8,-0.25) node [anchor=north] {$0.8$};
             \draw[black] (9,0) -- (9, -0.25);
             \draw (9,-0.25) node [anchor=north] {$0.9$};
             \draw[black] (10,0) -- (10, -0.25);
             \draw (10,-0.25) node [anchor=north] {$1$};

            \draw[black] (0,1) -- (-0.25,1);
            \draw (-0.25,1) node [anchor=east] {$0.1$};
             \draw[black] (0,2) -- (-0.25,2);
             \draw (-0.25,2) node [anchor=east] {$0.2$};
            \draw[black] (0,3) -- (-0.25,3);
            \draw (-0.25,3) node [anchor=east] {$0.3$};
             \draw[black] (0,4) -- (-0.25,4);
             \draw (-0.25,4) node [anchor=east] {$0.4$};
\draw[black] (0,5) -- (-0.25,5);
\draw (-0.25,5) node [anchor=east] {$0.5$};
             \draw[black] (0,6) -- (-0.25,6);
             \draw (-0.25,6) node [anchor=east] {$0.6$};
            \draw[black] (0,7) -- (-0.25,7);
            \draw (-0.25,7) node[anchor=east] {$0.7$};
             \draw[black] (0,8) -- (-0.25,8);
             \draw (-0.25,8) node [anchor=east] {$0.8$};
             \draw[black] (0,9) -- (-0.25,9);
             \draw (-0.25,9) node [anchor=east] {$0.9$};
             \draw[black] (0,10) -- (-0.25,10);
             \draw (-0.25,10) node [anchor=east] {$1$};

\end{tikzpicture}
}
   \caption{Choice beliefs of player 1}
    \label{fig:my_label}
\end{figure}
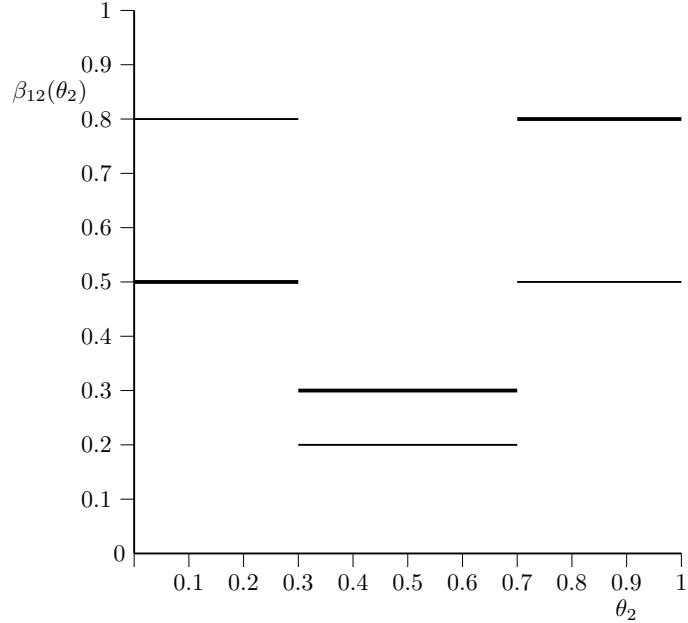

\[   
f'_{12}(\theta_2)=
     \begin{cases}
      \frac{2}{3}  & \text{if} \ 0 \leq \theta_2 \ \leq 0.3 \\
     \frac{1}{4} & \text{if}\ 0.3 \leq \theta_2 \ \leq 0.7 \\
       \frac{7}{3} & \text{if} \  0.7 \leq \theta_2 \ \leq 1\\ 
 
     \end{cases}
\]

\[   
f_{12}(\theta_2)=
     \begin{cases}
      \frac{1}{3}  & \text{if} \ 0 \leq \theta_2 \ \leq 0.3 \\
     \frac{7}{4} & \text{if}\ 0.3 \leq \theta_2 \ \leq 0.7 \\
       \frac{2}{3} & \text{if} \  0.7 \leq \theta_2 \ \leq 1\\ 
 
     \end{cases}
\]

Figure 1 plots the corresponding cumulative distribution function for both parameter beliefs, which shows that neither parameter belief stochastically dominates one another. The bold line corresponds to $f_{12}'$, and the other line to $f_{12}$. Next consider the following choice beliefs depicted in Figure 2, where again the bold lines correspond to $\beta_{12}'(\theta_2)$ and the other lines to $\beta_{12}(\theta_2)$.

\[   
\beta_{12}'(\theta_2)=
     \begin{cases}
      0.5  & \text{if} \ 0 \leq \theta_2 \ \leq 0.3 \\
     0.3 & \text{if}\ 0.3 \leq \theta_2 \ \leq 0.7 \\
       0.8 & \text{if} \  0.7 \leq \theta_2 \ \leq 1\\ 
 
     \end{cases}
\]

\[   
\beta_{12}(\theta_2)=
     \begin{cases}
      0.8  & \text{if} \ 0 \leq \theta_2 \ \leq 0.3 \\
     0.2 & \text{if}\ 0.3 \leq \theta_2 \ \leq 0.7 \\
       0.5 & \text{if} \  0.7 \leq \theta_2 \ \leq 1\\ 
 
     \end{cases}
\]

Neither choice belief is larger than the other for any $\theta_2$, nor are these choice beliefs always increasing or decreasing in the value of the parameter. However, we will now show that $(\beta_{12}' \circ f_{12}') \geq (\beta_{12} \circ f_{12})$:

\[   
\int_{c_2'}^{1} (\beta_{12}' \circ f_{12}')(c_{2}) \, d c_{2}=
     \begin{cases}
      0  & \text{if} \ c_2' > 0.8 \\
     0.7 & \text{if}\ 0.5 < c_2' \leq 0.8 \\
       0.9 & \text{if} \  0.3 <  c_2' \leq 0.5\\ 
       1 & \text{if} \  c_2' \leq 0.3\\ 
     \end{cases}
\]

\[   
\int_{c_2'}^{1} (\beta_{12} \circ f_{12})(c_{2}) \, d c_{2}=
     \begin{cases}
      0  & \text{if} \ c_2' > 0.8 \\
     0.1 & \text{if}\ 0.5 < c_2' \leq 0.8 \\
       0.3 & \text{if} \  0.2 <  c_2' \leq 0.5\\ 
       1 & \text{if} \  c_2' \leq 0.2\\ 
     \end{cases}
\]
Hence, the probability of high choices is greater for the former choice-parameter belief combination. Note that if $f_i' \geq f_i$ and $\beta_i' \geq \beta_i$, it is not necessarily true that $(\beta_i' \circ f_i') \geq (\beta_i \circ f_i)$. For example, if  both $\beta_i'$ and $ \beta_i$ are decreasing, and the probability of low parameter values is much greater for $f_i$, then it is not necessarily true that  $(\beta_i' \circ f_i') \geq (\beta_i \circ f_i)$. Lastly, if we consider the same parameter belief $f_i$ and $\beta_i' \geq \beta_i$. Then $(\beta_i' \circ f_i) \geq (\beta_i \circ f_i)$, but the opposite is not necessarily true.

\section{Monotonicity Results under Point Rationalizability}
We are interested in comparing two distinct parameter beliefs and choice beliefs, determining which one yields a greater optimal choice for the player. To this end, we assume that the expected utility function exhibits increasing or decreasing differences in parameters and the composite belief.

\begin{assumption}\label{increasingdiff}
Each player $i$'s expected utility function has weakly increasing differences. That is, for any $(\theta_i',\beta_i',f_i') \geq (\theta_i,\beta_i,f_i)$ and $c_i' \geq c_i $ it holds that

$$U_i(\theta_i',c_i',\beta_i',f_i')-U_i(\theta_i',c_i,\beta_i',f_i') \geq U_i(\theta_i,c_i',\beta_i,f_i)-U_i(\theta_i,c_i,\beta_i,f_i).$$
\end{assumption}
Similarly, player $i$'s expected utility function would have weakly decreasing differences if the inequality in Assumption \ref{increasingdiff} would be reversed. 
\newline
\newline
We next assume that there exists a unique optimal choice.
\begin{assumption}
For every choice belief $\beta_i$, parameter belief $f_i$ and parameter $\theta_i$, there is a unique $c_i \in C_i$ that is optimal for $(\theta_i,\beta_i,f_i)$.
\end{assumption}

For the next assumption, we assume a form of continuity in the game.
\begin{assumption}\label{continuity}
Consider $c_i < c_i'$ and let $c_i$ be optimal for player $i$ for $(\theta_i,\beta_i,f_i)$ and let $c_i'$ be optimal for $(\theta_i,\beta_i',f_i')$. Now consider some choice $c_i''$ such that $c_i< c_i'' <c_i'$. Then there exists some choice belief $\beta_i''$ and parameter belief $f_i''$ such that $c_i''$ is optimal for $(\theta_i,\beta_i'',f_i'')$, where $\beta_i''= (1-\lambda) \cdot \beta_i + \lambda \cdot \beta_i'$ and $f_i''= (1-\lambda) \cdot f_i + \lambda \cdot f_i'$, for some $\lambda \in [0,1]$ .
\end{assumption}

Lastly, we assume that for each player, there exists a parameter belief that first-order stochastically dominates all other parameter beliefs and that there exists a parameter belief that is first-order stochastically dominated by all other parameter beliefs.
\begin{assumption}
For each player $i$ there exists a $f_i',f_i'' \in M_i$ with $f_i' \geq f_i$ and $f_i'' \leq f_i$ for all $f_i \in M_i$.    
\end{assumption}

 Checking whether Assumption 1 holds might be a difficult task in practice. The following lemma shows that under certain conditions on the utility function of the players, Assumption 1 is implied.

 \begin{lemma}\label{increasingdifflemma}
Let $\frac{\partial^2u_i}{\partial c_i \partial \theta_i} \geq 0 $, $\frac{\partial^2u_i}{\partial c_i \partial c_j} \geq 0 \ \forall j \neq i$ and $\frac{\partial^3u_i}{\partial c_i \partial c_j \partial c_l} = 0 \ \forall j, l \neq i$ and $j\neq l$ . Then player $i$'s expected utility function has weakly increasing differences. Similarly, if $\frac{\partial^2u_i}{\partial c_i \partial \theta_i} \leq 0 $, $\frac{\partial^2u_i}{\partial c_i \partial c_j} \leq 0 \ \forall j \neq i$ and $\frac{\partial^3u_i}{\partial c_i \partial c_j \partial c_l} = 0 \ \forall j,l \neq i$, then player i's expected utility function has weakly decreasing differences.    
\end{lemma}

Lemma \ref{monotonicity} shows that if assumption 1 and 2 hold, then greater choices are optimal for greater composite beliefs and parameter values. 
\begin{lemma}\label{monotonicity}
 Let Assumption 1 and 2 hold. Consider player $i$ and $(\theta_i,\beta_i,f_i) \leq (\theta_i',\beta_i',f_i')$. If $c_i$ is optimal for $(\theta_i,\beta_i,f_i)$ and  $c_i'$ is optimal for $(\theta_i',\beta_i',f_i')$ , then $c_i \leq c_i'$.
\end{lemma}

Theorem 1 shows that for each player, the set of point rationalizable choices is an interval and that the greatest and least choice than can be chosen after round $k$ of the iterative procedure can be found by investigating the greatest and lowest composite belief. Furthermore, these lowest and greatest choices are both weakly increasing in the value of the parameter.

\begin{theorem}
 Let Assumptions 1,2,3 and 4 hold and consider player $i$. Then $P_i^k(\theta_i)=[l_i^k(\theta_i),u_i^k(\theta_i)]$, where $l_i^k(\theta_i)$ and $u_i^k(\theta_i)$ denote the smallest and greatest choice that player $i$ with parameter $\theta_i$ can choose after round $k$ of the iterative procedure.  If $\theta_i < \theta_i'$, then $l_i^k(\theta_i) \leq l_i^k(\theta_i')$ and $u_i^k(\theta_i) \leq u_i^k(\theta_i')$ for every $k \in \mathbb{N}$. Moreover,
 $$l_i^k(\theta_i)=\{  c_i \in P_i^{k-1}(\theta_i) | \ c_i \ \text{is optimal in} \ P_i^{k-1}(\theta_i) \ \text{for the choice belief} \ \beta_i'' \ \text{such that}$$
$$  \ \beta_{ij}''(\theta_j)= \min P_j^{k-1}(\theta_j) \ \forall \theta_j \in \Theta_j \ \text{and} \ j \neq i \ \text{and a parameter belief} \ f_i'' \ \text{such that} \ f_i'' \leq f_i \ \forall f_i \in M \}.$$
Similarly, we have that 
$$u_i^k(\theta_i)=\{ c_i \in P_i^{k-1}(\theta_i) | \ c_i \ \text{is optimal in} \ P_i^{k-1}(\theta_i) \ \text{for the choice belief} \ \beta_i' \ \text{such that}$$
$$  \ \beta_{ij}'(\theta_j)= \max P_j^{k-1}(\theta_j) \ \forall \theta_j \in \Theta_j \ \text{and} \ j \neq i \ \text{and a parameter belief} \ f_i' \ \text{such that} \ f_i' \geq f_i \ \forall f_i \in M \}$$
Furthermore, we have that $P_i(\theta_i)=[ l_i(\theta_i),u_i(\theta_i)]$, where $l_i(\theta_i) \leq l_i(\theta_i') \ \text{and} \ u_i(\theta_i) \leq u_i(\theta_i')$ whenever $\theta_i < \theta_i'$.
 \end{theorem}
 Theorem 1 shows that if a player would gain new information that changes the set of parameter beliefs, then the set of point rationalizable choices does not change as long as the highest and the lowest parameter belief of this player remains the same. 
 
 \section{Applications}
 \subsection{Bertrand Model with Differentiated Goods}
Consider a Bertrand model with two firms and differentiated goods with the following profit functions. 
 $$\pi_1(\theta_1,p_1,p_2)=(p_1-\theta_{1})(a- p_1+ p_2) \ \text{and}$$
 
 $$\pi_2(\theta_2,p_2,p_1)=(p_2-\theta_{2})(a- p_2+ p_1),$$
 where $p_1,p_2$ are the prices that the firms choose, and $\theta_1,\theta_2$ are the marginal costs of the firm.
 Analyzing the profit functions, we observe that a firm's received demand is positively correlated with the opponents' price and negatively correlated with its own price. When a firm increases its price, it experiences a gradual decline in demand rather than an immediate drop. Moreover, firms have knowledge of their own marginal costs but face uncertainty regarding their opponents' marginal costs. However, they do have certainty that the opponents' costs fall within a specified range or domain. Let $\Theta_1=\Theta_2=[0, \phi]$. Let $M_i=\{f_i^{\alpha} | \alpha \in [0,\frac{2}{\phi}] \}$, where 

 \[   
f_{ij}^{\alpha}(\theta_j)=
     \begin{cases}
      \alpha  & \text{if} \ \theta_j \in [0,\frac{\phi}{2}] \\
     \frac{2}{\phi}- \alpha & \text{if}\ \theta_j \in (\frac{\phi}{2},\phi] \\
     \end{cases}
\]
 Let each firm have a choice set of $C_1=C_2=[0,\Bar{p}]$. We have that $\frac{\partial^2\pi_1}{\partial p_1 \partial \theta_{1}}= 1 >0$, and $\frac{\partial^2\pi_1}{\partial p_1 \partial p_2}=1 \geq 0$, and similarly for firm 2. The third order condition for Lemma \ref{increasingdifflemma} is automatically satisfied as this condition is only necessary to check for three or more players. Hence, by Lemma \ref{increasingdifflemma} the expected utility function of each firm has weakly increasing differences and Assumption 1 holds. 
 \newline
 \newline
 Consider firm 1 with valuation $\theta_1$, choice belief $\beta_1$ and parameter belief $f_1^{\alpha}$. The expected utility function of firm 1 can be denoted as
 
   $$U_1(\theta_1,p_1,\beta_1,f_1^{\alpha})=\alpha \cdot \int_0^{\frac{\phi}{2}}  (p_1-\theta_{1})(a- p_1+\beta_{12}(\theta_{2})) \, d\theta_{2} + (\frac{2}{\phi}-\alpha) \cdot \int_{\frac{\phi}{2}}^{\phi} (p_1-\theta_{1})(a- p_1+ \beta_{12}(\theta_{2})) \, d\theta_{2} $$
   
    $$= (p_1-\theta_1)(a-p_1 + \alpha \cdot  \int_0^{\frac{\phi}{2}} \beta_{12}(\theta_{2})  \, d\theta_{2} + (\frac{2}{\phi}-\alpha) \int_{\frac{\phi}{2}}^{\phi} \beta_{12}(\theta_{2})  \, d\theta_{2}  )$$
    
Hence, we can rewrite the expected utility function of firm 1 to

$$U_1(\theta_1,p_1,\beta_1,f_1^{\alpha})= (p_1-\theta_{1})(a- p_1+\E (p_2 \ |\beta_1,f_1^{\alpha}))$$

where $ \E (p_2 \ |\beta_1,f_1^{\alpha})=  \alpha \cdot  \int_0^{\frac{\phi}{2}} \beta_{12}(\theta_{2})  \, d\theta_{2} + (\frac{2}{\phi}-\alpha) \int_{\frac{\phi}{2}}^{\phi} \beta_{12}(\theta_{2})  \, d\theta_{2}  $. Hence, for any $(\theta_1, \beta_1, f_1^{\alpha})$, the unique optimal choice of firm 1 is given by $p_1= \frac{a + \theta_{1}}{2}+ \frac{1}{2} \cdot \E (p_2 \ |\beta_1,f_1^{\alpha})$. As a result, Assumption 2 holds. 
\newline
\newline
To show that Assumption 3 holds, let $p_1$ be optimal for $(\theta_1, \beta_1,f_1)$ and let $p_1' > p_1$ be optimal for $(\theta_1,\beta_1',f_1')$. Now consider some $p_1 < p_1'' < p_1'$. Consider the function $g(\lambda)=\frac{a+\theta_1}{2} + \frac{1}{2}\E(p_2 \ |\lambda \cdot \beta_1 + (1-\lambda) \cdot \beta_1', \lambda \cdot f_1 + (1-\lambda) \cdot f_1')$. Because $g$ is continuous, by the intermediate value theorem, as $g(0) = p_1'$ and $g(1)= p_1,$ there exists some $\lambda \in (0,1)$ such that $g(\lambda)= p_1''$. Hence, $p_1''$ is optimal for $(\theta_1,\lambda \cdot \beta_1 + (1-\lambda) \cdot \beta_1', \lambda \cdot f_1 + (1-\lambda) \cdot f_1')$.
\newline     
\newline   
Finally, Assumption 4 holds as well as $f_i^{0} \geq f_i^{\alpha} \ \forall \alpha \in [0,\frac{2}{\phi}]$ and $f_i^{\frac{2}{\phi}} \leq f_i^{\alpha} \ \forall \alpha \in [0,\frac{2}{\phi}]$. 
\newline
\newline
Theorem 1 then implies that we can find the lower and upper bound after round 1 of the iterative procedure by investigating the lowest possible price that firm 2 can choose together with the lowest possible parameter belief $f_1^{\frac{\phi}{2}}$ and the highest possible price that firm 2 can choose together with the highest possible parameter belief $f_1^{0}$. We will show that for every $k \geq 1$, the set of choices of a firm with parameter $\theta_i$ that survive round $k$ of the procedure can be denoted as

$$P_i^k(\theta_i)=[(1-\frac{1}{2^k})a + (\frac{1}{8}-\frac{1}{2^{k+2}})\phi + \frac{1}{2}\theta_i,(1-\frac{1}{2^k})a + (\frac{3}{8}-\frac{3}{2^{k+2}})\phi + \frac{1}{2^k}\Bar{p} + \frac{1}{2}\theta_i ]$$ 

Consider round 1 of the iterative procedure. Consider any parameter belief and the highest choice belief $\beta_1^{**}$, where $\beta_{12}^{**}(\theta_2)=\Bar{p} \ \forall \theta_2 \in \Theta_2$. Then $\E (p_2 \ |\beta_1^{**},f_1^{\alpha})=\Bar{p}$ and 
the optimal price of firm 1 is then $p_1=\frac{a + \theta_{1}}{2}+ \frac{1}{2} \cdot \E (p_2 \ |\beta_1^{**},f_1^{\alpha})=\frac{a+\Bar{p}+\theta_1}{2}$
\newline
\newline
Similarly, when we consider consider the lowest possible choice belief $\beta_1^*$, where $\beta_{12}^*(\theta_2)=0 \ \forall \theta_2 \in \Theta_2$, then this leads to an optimal price of $\frac{a+\theta_{1}}{2}$.
\newline
\newline
Hence, we have that $P_1^1(\theta_1)=[\frac{a+\theta_{1}}{2},\frac{a+\Bar{p}+\theta_{1}}{2}]$ and similarly $P_2^1(\theta_2)=[\frac{a+\theta_{2}}{2},\frac{a+\Bar{p}+\theta_{2}}{2}]$.
\newline
\newline
Now consider round $k \geq 2$ of the iterative procedure and assume that 
$$P_1^{k-1}(\theta_i)=[(1-\frac{1}{2^{k-1}})a + (\frac{1}{8}-\frac{1}{2^{k+1}})\phi + \frac{1}{2}\theta_1,(1-\frac{1}{2^{k-1}})a + (\frac{3}{8}-\frac{3}{2^{k+1}})\phi + \frac{1}{2^{k-1}}\Bar{p} + \frac{1}{2}\theta_1 ]$$ 

and similarly for firm 2. 
\newline
\newline
Consider firm 1. The highest parameter belief of firm 1 is $f_1^{0}$ and the highest choice belief is  $\beta_1^{**}$, where $\beta_{12}^{**}(\theta_2)=(1-\frac{1}{2^{k-1}})a + (\frac{3}{8}-\frac{3}{2^{k+1}})\phi + \frac{1}{2^{k-1}}\Bar{p} + \frac{1}{2}\theta_2 \ \forall \theta_2 \in \Theta_2$. Then $\E (p_2 \ |\beta_1^{**},f_1^0)=(1-\frac{1}{2^{k-1}})a + (\frac{3}{8}-\frac{3}{2^{k+1}})\phi + \frac{1}{2^{k-1}}\Bar{p} + \frac{3}{8} \phi$. The optimal choice of firm 1 is then given by $p_1= \frac{a+\theta_1}{2}+ \frac{1}{2} \E (p_2 \ |\beta_1^{**},f_1^0)= \frac{a+\theta_1}{2}+ (\frac{1}{2}-\frac{1}{2^{k}})a + (\frac{3}{16}-\frac{3}{2^{k+2}})\phi + \frac{1}{2^{k}}\Bar{p} + \frac{3}{16} \phi= (1-\frac{1}{2^k})a + (\frac{3}{8}-\frac{3}{2^{k+2}})\phi + \frac{1}{2^k}\Bar{p} + \frac{1}{2}\theta_1 $. The proof of the lower bound is very similar, which completes the proof. 
\newline
\newline
Finally, we are able to investigate the set of point rationalizable choices. We have that

$$P_i(\theta_i)=[a +\frac{1}{8} \phi + \frac{1}{2}\theta_i,a + \frac{3}{8}\phi + \frac{1}{2}\theta_i ]$$

Both the upper bound and lower bound of the set of point rationalizable choices of a firm is increasing in $a$, $\phi$ and $\theta_i$. An increase in $\theta_i$ results in a higher price for firm $i$ as it needs to compensate for the rise in costs. Furthermore, this is true because of Theorem 1. Similarly, an increase in $\phi$ leads to a higher price because it implies that firm $i$ anticipates firm $j$ choosing a higher price, allowing firm $i$ to raise its own price accordingly.

\subsection{Cournot Model with Demand Uncertainty}
Consider a Cournot model with two firms with the following profit functions. Instead of choosing a price, in this model the firms choose a quantity they will sell. The profit functions of firm 1 and 2 are given by:

$$\pi_1(\theta_1,q_1,q_2)= (a-\theta_{1}(q_1+q_2)-c)q_1$$

$$\pi_2(\theta_2,q_2,q_1)= (a-\theta_{2}(q_1+q_2)-c)q_2$$
In this model, the firms know both have the same constant marginal costs given by $c$, but they are uncertain how the demand will change when total quantity changes. The parameters  $\theta_{1}$ and $\theta_{2}$ represent a belief of firms 1 and 2 about the price elasticity of demand. Let $\Theta_1=\Theta_2=[\underline{\phi}, \Bar{\phi}]$. Let $M_i=\{f_i^{\alpha} | \alpha \in [0,\frac{2}{\Bar{\phi}-\underline{\phi}}] \}$, where 

 \[   
f_{ij}^{\alpha}(\theta_j)=
     \begin{cases}
      \alpha  & \text{if} \ \theta_j \in [\underline{\phi},\frac{\Bar{\phi}+\underline{\phi}}{2}] \\
     \frac{2}{\Bar{\phi}-\underline{\phi}}- \alpha & \text{if}\ \theta_j \in (\frac{\Bar{\phi}+\underline{\phi}}{2},\Bar{\phi}] \\
     \end{cases}
\]
 Let each firm have a choice set of $C_1=C_2=[0,\Bar{q}]$. Contrary to the Bertrand example, we now have that $\frac{\partial^2\pi_1}{\partial q_1 \partial \theta_{1}}=-2q_1-q_2 \leq 0$, and $\frac{\partial^2\pi_1}{\partial q_1 \partial q_2}=-\theta_{1} \leq 0$, and similarly for firm 2. Similarly to the Bertrand example, the third order condition for Lemma \ref{increasingdifflemma} is automatically satisfied. Hence, by Lemma \ref{increasingdifflemma}  the expected utility function of each firm has weakly decreasing differences and Assumption 1 holds.
\newline
\newline
The expected utility of firm 1 with parameter $\theta_1$, parameter belief $f_1^{\alpha}$ and choice belief $\beta_{1}$ can then be denoted as 

$$U_1(\theta_1,q_1,\beta_1,f_1^{\alpha})= \alpha \cdot q_1 \cdot  \int_{\underline{\phi}}^{\frac{\Bar{\phi}+\underline{\phi}}{2}}(a-\theta_{1} \cdot q_1- \theta_{1} \cdot \beta_{12}(\theta_{2})-c) \, d\theta_{2}$$ 
$$+ ( \frac{2}{\Bar{\phi}-\underline{\phi}}- \alpha) \cdot q_1 \cdot  \int_{\frac{\Bar{\phi}+\underline{\phi}}{2}}^{\Bar{\phi}}(a-\theta_{1} \cdot q_1- \theta_{1} \cdot \beta_{12}(\theta_{2})-c) \, d\theta_{2}$$

Hence, we can rewrite the expected utility function of firm 1 to

$$U_1(\theta_1,q_1,\beta_1,f_1^{\alpha})= q_1 \cdot (a-\theta_{1} \cdot q_1- \theta_{1} \cdot \E (q_2 \ | \beta_1,f_1^{\alpha})-c)$$

where $ \E (q_2 \ | \beta_1,f_1^{\alpha})= \alpha \cdot  \int_{\underline{\phi}}^{\frac{\Bar{\phi}+\underline{\phi}}{2}} \beta_{12}(\theta_{2})  \, d\theta_{2} +  ( \frac{2}{\Bar{\phi}-\underline{\phi}}- \alpha)  \int_{\frac{\Bar{\phi}+\underline{\phi}}{2}}^{\Bar{\phi}} \beta_{12}(\theta_{2})  \, d\theta_{2} $. Hence, for any $(\theta_1,\beta_1,f_1^\alpha)$, there exists a unique optimal quantity $q_1= \frac{a-c}{2\theta_{1}} - \frac{1}{2} \cdot \E(q_2 \ | \beta_1,f_1^{\alpha}$ and Assumption 2 holds. 
\newline
\newline
To show that Asssumption 3 holds, let $q_1$ be optimal for $(\theta_1, \beta_1,f_1^\alpha)$ and let $q_1' > q_1$ be optimal for $(\theta_1,\beta_1',f_1^{\alpha'})$. Now consider some $q_1 < q_1'' < q_1'$. Consider the function $g(\lambda)=\frac{a-c}{2\theta_{1}} - \frac{1}{2} \cdot \E(p_2 \ | \lambda \cdot \beta_1 + (1-\lambda) \cdot \beta_1', \lambda \cdot f_1^\alpha + (1-\lambda) \cdot f_1^{\alpha'})$. As $g$ is continuous, by the intermediate value theorem, as $g(0) = q_1'$ and $g(1)= q_1,$ there exists a $\lambda \in (0,1)$ such that $g(\lambda)= q_1''$. Hence, $q_1''$ is optimal for $(\theta_1,\lambda \cdot \beta_1 + (1-\lambda) \cdot \beta_1', \lambda \cdot f_1^\alpha + (1-\lambda) \cdot f_1^{\alpha'}) $
\newline
\newline
Finally, Assumption 4 holds as well as $f_i^{0} \geq f_i^{\alpha} \ \forall \alpha \in [0,\frac{2}{\Bar{\phi}-\underline{\phi}}]$ and $f_i^{\frac{2}{\Bar{\phi}-\underline{\phi}}} \leq f_i^{\alpha} \ \forall \alpha \in [0,\frac{2}{\Bar{\phi}-\underline{\phi}}]$. 
\newline
\newline
We will show that if $k$ is even, then 

$$P_1^k(\theta_1)=[\frac{a-c}{2\theta_{1}} - \frac{1}{3}(a-c) \frac{1}{\Bar{\phi}-\underline{\phi}}(2\ln(\frac{\Bar{\phi}+\underline{\phi}}{2\underline{\phi}})(1-(\frac{1}{4})^{\frac{k}{2}}) - \ln (\frac{2 \Bar{\phi}}{\Bar{\phi} +\underline{\phi}})(1-(\frac{1}{4})^{\frac{k}{2}-1})  )  ,$$ $$\frac{a-c}{2\theta_{1}} - \frac{1}{3}(a-c) \frac{1}{\Bar{\phi}-\underline{\phi}}(2\ln (\frac{2 \Bar{\phi}}{\Bar{\phi} +\underline{\phi}})(1-(\frac{1}{4})^{\frac{k}{2}}) - \ln(\frac{\Bar{\phi}+\underline{\phi}}{2\underline{\phi}})(1-(\frac{1}{4})^{\frac{k}{2}-1})  )  + \frac{\Bar{q}}{2^k} ]$$

Similarly if $k$ is odd

$$P_1^k(\theta_1)=[\frac{a-c}{2\theta_{1}} - \frac{1}{3}(a-c) \frac{1}{\Bar{\phi}-\underline{\phi}}(2\ln(\frac{\Bar{\phi}+\underline{\phi}}{2\underline{\phi}})- \ln (\frac{2 \Bar{\phi}}{\Bar{\phi} +\underline{\phi}}))(1-(\frac{1}{4})^{\frac{k-1}{2}}) - \frac{\Bar{q}}{2^k},$$ $$\frac{a-c}{2\theta_{1}} - \frac{1}{3}(a-c) \frac{1}{\Bar{\phi}-\underline{\phi}}(2\ln (\frac{2 \Bar{\phi}}{\Bar{\phi} +\underline{\phi}})- \ln(\frac{\Bar{\phi}+\underline{\phi}}{2\underline{\phi}}))(1-(\frac{1}{4})^{\frac{k-1}{2}}) ]$$

Consider round 1 of the iterative procedure. Similarly to the Bertrand example, consider any parameter belief and the lowest possible choice belief $\beta_1^*$, where $\beta_{12}^*(\theta_2)= 0 \ \forall \theta_2 \in \Theta_2$. This leads to a corresponding optimal quantity of $q_1= \frac{a-c}{2\theta_{1}}$ and considering the highest possible choice belief $\beta_1^{**}$, where $\beta_{12}^{**}(\theta_2)=\Bar{q} \ \forall \theta_2 \in \Theta_2$ leads to an optimal quantity of $q_1= \frac{a-c}{2\theta_1}-\frac{\Bar{q}}{2}$. Hence, we have 

$$P_1^1(\theta_1)=[\frac{a-c}{2\theta_{1}} - \frac{\Bar{q}}{2},\frac{a-c}{2\theta_{1}}].$$
Now consider round 2 of the iterative procedure. We consider the lowest choice belief $\beta_1^*$, where $\beta_{12}^*(\theta_2)=\frac{a-c}{2\theta_{2}} - \frac{\Bar{q}}{2} \ \forall \theta_2 \in \Theta_2$. Note that the greater the value of $\theta_2$, the lower $\beta_{12}^*(\theta_2)$ is, which is why the highest parameter belief $f_1^0$ is considered. The corresponding optimal quantity of firm 1 is then $q_1= \frac{a-c}{2\theta_{1}} - \frac{1}{2} \cdot \E(q_2 \ |\beta_1^*,f_1^0)$, where $\E(q_2 \ |\beta_1^*,f_1^0)= \frac{2}{\Bar{\phi}- \underline{\phi}}\int_{\frac{\Bar{\phi}+\underline{\phi}}{2}}^{\Bar{\phi}} ( \frac{a-c}{2\theta_{2}} - \frac{\Bar{q}}{2})  \, d\theta_{2}= \frac{1}{\Bar{\phi}-\underline{\phi}} \cdot \ln (\frac{2 \Bar{\phi}}{\Bar{\phi} + \underline{\phi}}) \cdot (a-c) - \frac{\Bar{q}}{2}$. Hence $q_1= \frac{a-c}{2\theta_{1}} - \frac{1}{\Bar{\phi}-\underline{\phi}} \ln (\frac{2 \Bar{\phi}}{\Bar{\phi} +\underline{\phi}}) \cdot (\frac{a-c}{2} ) + \frac{\Bar{q}}{4}  $. Similarly, considering the lowest possible parameter belief $f_1^{\frac{2}{\Bar{\phi}-\underline{\phi}}}$ and highest choice belief $\beta_1^{**}$, where $\beta_{12}^{**}(\theta_2)=\frac{a-c}{2\theta_{2}} \ \forall \theta_2 \in \Theta_2$ leads to an optimal quantity of $q_1=  \frac{a-c}{2\theta_{1}} - \frac{1}{\Bar{\phi}-\underline{\phi}} \ln (\frac{\Bar{\phi}+\underline{\phi}}{2 \underline{\phi}}) \cdot (\frac{a-c}{2} ) $. Hence, we have 

$$P_1^2(\theta_1)=[\frac{a-c}{2\theta_{1}} - \frac{1}{\Bar{\phi}-\underline{\phi}} \ln (\frac{\Bar{\phi}+\underline{\phi}}{2 \underline{\phi}}) \cdot (\frac{a-c}{2} ), \frac{a-c}{2\theta_{1}} - \frac{1}{\Bar{\phi}-\underline{\phi}} \ln (\frac{2 \Bar{\phi}}{\Bar{\phi} +\underline{\phi}}) \cdot (\frac{a-c}{2} ) + \frac{\Bar{q}}{4}]$$

Let $k \geq 2$ and let the statement be true for the first $k-1$ rounds, where $k-1$ is odd. Then 

$$P_2^{k-1}(\theta_2)=[\frac{a-c}{2\theta_{2}} - \frac{1}{3}(a-c) \frac{1}{\Bar{\phi}-\underline{\phi}}(2\ln(\frac{\Bar{\phi}+\underline{\phi}}{2\underline{\phi}})- \ln (\frac{2 \Bar{\phi}}{\Bar{\phi} +\underline{\phi}}))(1-(\frac{1}{4})^{\frac{k}{2}-1}) - \frac{\Bar{q}}{2^{k-1}},$$ $$\frac{a-c}{2\theta_{2}} - \frac{1}{3}(a-c) \frac{1}{\Bar{\phi}-\underline{\phi}}(2\ln (\frac{2 \Bar{\phi}}{\Bar{\phi} +\underline{\phi}})- \ln(\frac{\Bar{\phi}+\underline{\phi}}{2\underline{\phi}}))(1-(\frac{1}{4})^{\frac{k}{2}-1}) ]$$

Consider round $k$ and firm 1. Consider the highest parameter belief  $f_1^0$ and the lowest choice belief $\beta_1^{**}$, where $\beta_{12}^{**}(\theta_2)= \frac{a-c}{2\theta_{2}} - \frac{1}{3}(a-c) \frac{1}{\Bar{\phi}-\underline{\phi}}(2\ln(\frac{\Bar{\phi}+\underline{\phi}}{2\underline{\phi}})- \ln (\frac{2 \Bar{\phi}}{\Bar{\phi} +\underline{\phi}}))(1-(\frac{1}{4})^{\frac{k}{2}-1}) - \frac{\Bar{q}}{2^{k-1}}  \ \forall \theta_2 \in \Theta_2$.

Then $\E(\beta_1^{**},f_1^0)= \frac{1}{\Bar{\phi}-\underline{\phi}} \cdot \ln (\frac{2 \Bar{\phi}}{\Bar{\phi}+ \underline{\phi}}) \cdot (a-c) - \frac{1}{3}(a-c) \frac{1}{\Bar{\phi}-\underline{\phi}}(2\ln(\frac{\Bar{\phi}+\underline{\phi}}{2\underline{\phi}})- \ln (\frac{2 \Bar{\phi}}{\Bar{\phi} +\underline{\phi}}))(1-(\frac{1}{4})^{\frac{k}{2}-1}) - \frac{\Bar{q}}{2^{k-1}} $. Hence, the optimal quantity is

$$q_1= \frac{a-c}{2\theta_1} - \frac{1}{2} \frac{1}{\Bar{\phi}-\underline{\phi}} \cdot \ln (\frac{2 \Bar{\phi}}{\Bar{\phi}+ \underline{\phi}}) \cdot (a-c) + \frac{1}{6}(a-c) \frac{1}{\Bar{\phi}-\underline{\phi}}(2\ln(\frac{\Bar{\phi}+\underline{\phi}}{2\underline{\phi}})- \ln (\frac{2 \Bar{\phi}}{\Bar{\phi} +\underline{\phi}}))(1-4(\frac{1}{4})^{\frac{k}{2}}) + \frac{\Bar{q}}{2^{k}}  $$.  

$$= \frac{a-c}{2\theta_{1}} - \frac{1}{3}(a-c) \frac{1}{\Bar{\phi}-\underline{\phi}}(2\ln (\frac{2 \Bar{\phi}}{\Bar{\phi} +\underline{\phi}})(1-(\frac{1}{4})^{\frac{k}{2}}) - \ln(\frac{\Bar{\phi}+\underline{\phi}}{2\underline{\phi}})(1-(\frac{1}{4})^{\frac{k}{2}-1})  )  + \frac{\Bar{q}}{2^k}$$
The proof for the lower bound is very similar. As a result we have that 

$$P_1^k(\theta_1)=[\frac{a-c}{2\theta_{1}} - \frac{1}{3}(a-c) \frac{1}{\Bar{\phi}-\underline{\phi}}(2\ln(\frac{\Bar{\phi}+\underline{\phi}}{2\underline{\phi}})(1-(\frac{1}{4})^{\frac{k}{2}}) - \ln (\frac{2 \Bar{\phi}}{\Bar{\phi} +\underline{\phi}})(1-(\frac{1}{4})^{\frac{k}{2}-1})  )  ,$$ $$\frac{a-c}{2\theta_{1}} - \frac{1}{3}(a-c) \frac{1}{\Bar{\phi}-\underline{\phi}}(2\ln (\frac{2 \Bar{\phi}}{\Bar{\phi} +\underline{\phi}})(1-(\frac{1}{4})^{\frac{k}{2}}) - \ln(\frac{\Bar{\phi}+\underline{\phi}}{2\underline{\phi}})(1-(\frac{1}{4})^{\frac{k}{2}-1})  )  + \frac{\Bar{q}}{2^k} ]$$

Now let the statement be true for the first $k-1$ rounds, where $k-1$ is even. Then 

$$P_2^{k-1}(\theta_2)=[\frac{a-c}{2\theta_{2}} - \frac{1}{3}(a-c) \frac{1}{\Bar{\phi}-\underline{\phi}}(2\ln(\frac{\Bar{\phi}+\underline{\phi}}{2\underline{\phi}})(1-(\frac{1}{4})^{\frac{k-1}{2}}) - \ln (\frac{2 \Bar{\phi}}{\Bar{\phi} +\underline{\phi}})(1-(\frac{1}{4})^{\frac{k-1}{2}-1}))  ,$$ $$\frac{a-c}{2\theta_{2}} - \frac{1}{3}(a-c) \frac{1}{\Bar{\phi}-\underline{\phi}}(2\ln (\frac{2 \Bar{\phi}}{\Bar{\phi} +\underline{\phi}})(1-(\frac{1}{4})^{\frac{k-1}{2}}) - \ln(\frac{\Bar{\phi}+\underline{\phi}}{2\underline{\phi}})(1-(\frac{1}{4})^{\frac{k-1}{2}-1})  )  + \frac{\Bar{q}}{2^{k-1}} ]$$

Let us now consider the lowest parameter belief $f_1^{\frac{2}{\Bar{\phi}-\underline{\phi}}}$ and the highest choice belief $\beta_1^{*}$, where $\beta_{12}^{*}(\theta_2)=\frac{a-c}{2\theta_{2}} - \frac{1}{3}(a-c) \frac{1}{\Bar{\phi}-\underline{\phi}}(2\ln (\frac{2 \Bar{\phi}}{\Bar{\phi} +\underline{\phi}})(1-(\frac{1}{4})^{\frac{k-1}{2}}) - \ln(\frac{\Bar{\phi}+\underline{\phi}}{2\underline{\phi}})(1-(\frac{1}{4})^{\frac{k-1}{2}-1})  )  + \frac{\Bar{q}}{2^{k-1}}  \ \forall \theta_2 \in \Theta_2$. Then $\E(\beta_1^{*},f_1(\phi))=  \frac{1}{\Bar{\phi}-\underline{\phi}} \cdot \ln (\frac{\Bar{\phi}+\underline{\phi}}{2\underline{\phi}}) \cdot (a-c) - \frac{1}{3}(a-c) \frac{1}{\Bar{\phi}-\underline{\phi}}(2\ln (\frac{2 \Bar{\phi}}{\Bar{\phi} +\underline{\phi}})(1-(\frac{1}{4})^{\frac{k-1}{2}}) - \ln(\frac{\Bar{\phi}+\underline{\phi}}{2\underline{\phi}})(1-(\frac{1}{4})^{\frac{k-1}{2}-1})  )  + \frac{\Bar{q}}{2^{k-1}} $. Hence,

$$q_1= \frac{a-c}{2\theta_1} - \frac{1}{2} \frac{1}{\Bar{\phi}-\underline{\phi}} \cdot \ln (\frac{\Bar{\phi}+\underline{\phi}}{2\underline{\phi}}) \cdot (a-c) +  \frac{1}{6}(a-c) \frac{1}{\Bar{\phi}-\underline{\phi}}(2\ln (\frac{2 \Bar{\phi}}{\Bar{\phi} +\underline{\phi}})(1-(\frac{1}{4})^{\frac{k-1}{2}}) - \ln(\frac{\Bar{\phi}+\underline{\phi}}{2\underline{\phi}})(1-(\frac{1}{4})^{\frac{k-1}{2}-1})  ) - \frac{\Bar{q}}{2^{k}} $$

$$= \frac{a-c}{2\theta_{1}} - \frac{1}{3}(a-c) \frac{1}{\Bar{\phi}-\underline{\phi}}(2\ln(\frac{\Bar{\phi}+\underline{\phi}}{2\underline{\phi}})- \ln (\frac{2 \Bar{\phi}}{\Bar{\phi} +\underline{\phi}}))(1-(\frac{1}{4})^{\frac{k-1}{2}}) - \frac{\Bar{q}}{2^k} $$

The proof for the lower bound is very similar. As a result we have that 

$$P_1^k(\theta_1)=[\frac{a-c}{2\theta_{1}} - \frac{1}{3}(a-c) \frac{1}{\Bar{\phi}-\underline{\phi}}(2\ln(\frac{\Bar{\phi}+\underline{\phi}}{2\underline{\phi}})- \ln (\frac{2 \Bar{\phi}}{\Bar{\phi} +\underline{\phi}}))(1-(\frac{1}{4})^{\frac{k-1}{2}}) - \frac{\Bar{q}}{2^k},$$ $$\frac{a-c}{2\theta_{1}} - \frac{1}{3}(a-c) \frac{1}{\Bar{\phi}-\underline{\phi}}(2\ln (\frac{2 \Bar{\phi}}{\Bar{\phi} +\underline{\phi}})- \ln(\frac{\Bar{\phi}+\underline{\phi}}{2\underline{\phi}}))(1-(\frac{1}{4})^{\frac{k-1}{2}}) ]$$

The set of point rationalizable choices of agent 1 is thus

$$P_1(\theta_1)=[\frac{a-c}{2\theta_{1}} - \frac{1}{3}(a-c) \frac{1}{\Bar{\phi}-\underline{\phi}}(2\ln(\frac{\Bar{\phi}+\underline{\phi}}{2\underline{\phi}}) - \ln (\frac{2 \Bar{\phi}}{\Bar{\phi} +\underline{\phi}}))  ,\frac{a-c}{2\theta_{1}} - \frac{1}{3}(a-c) \frac{1}{\Bar{\phi}-\underline{\phi}}(2\ln (\frac{2 \Bar{\phi}}{\Bar{\phi} +\underline{\phi}})- \ln(\frac{\Bar{\phi}+\underline{\phi}}{2\underline{\phi}})) ]$$

The upper and lower bound of the set of point rationalizable choices of each firm is increasing in $a$, and decreasing in $\theta_i$ and $c$, which is also true because of Theorem 1. An increase in $\theta_i$ results in a greater price elasticity of demand, which makes producing less more favorable. Note that in both examples, if both players would have only one parameter belief, then the set of point rationalizable choices of a player would consist of a single choice.

\section{Discussion and Conclusion}
A limitation in this paper is the one dimensionality of the choice set and parameter set of the players. In future work, the model could be extended to a multidimensional choice set and parameter set. By limiting the analysis to a one-dimensional setting, we can keep the model and analysis simple and easy to understand. Introducing additional dimensions to the choice set and parameter set could significantly increase the complexity of the model and analysis, making it harder to derive clear and intuitive results.
\newline
\newline
Finally, it is worth noting that the assumption of a third-order partial derivative of zero with respect to his choice and any two choices of his opponents may rule out some interesting scenarios, such as public good environments. Although this assumption is sufficient to ensure weakly increasing differences for the expected utility function, it may be refined to a condition that is necessary for weakly increasing differences. However, it is important to highlight that this assumption is often satisfied in many industrial organization models, such as the Cournot and Bertrand models. Nevertheless, further research could explore the relaxation or generalization of this assumption to accommodate a wider range of strategic environments.
\newline
\newline

\newpage
\section{Appendix}
We need the following Lemma to prove Lemma \ref{increasingdifflemma}. It is a generalization of \citeA{mas1995microeconomic} from one-dimensional variables to multi-dimensional variables.
    
\begin{lemma}\label{mass colel}
 Let $X_j=[\underline{x}_j,\Bar{x}_j]$ for every $j \in \{1,...,n\}$. Let $u= X_1 \times...\times X_n \to \mathbb{R}$ where $\frac{\partial u(x_1,...,x_n) }{\partial x_j} \geq 0$ for all $j \in \{1,..,n\}$ and $\frac{\partial^2 u(x_1,...,x_n) }{\partial x_j \partial x_k}=0$ for all $j\neq k \in \{1,...,n\}$. For every $j \in \{1,...,n\}$, let $f_j,g_j$ be probability density functions on $X_j$ and $F_j$,$G_j$ be the cumulative distribution functions belonging to $f_j,g_j$, where $F_j$ first-order stochastically dominates $G_j$. Then 

$$\int_{\underline{x}_1}^{\Bar{x}_1} \dots \int_{\underline{x}_n}^{\Bar{x}_n}u (x_1,...,x_{n}) \cdot  \prod_{j=1 }^n f_j(x_j) \,dx_{n} \dots dx_1$$ $$ \geq \int_{\underline{x}_1}^{\Bar{x}_1} \dots \int_{\underline{x}_n}^{\Bar{x}_n}u (x_1,...,x_{n}) \cdot  \prod_{j=1 }^n g_j(x_j)  \,dx_{n} \dots dx_1$$
\end{lemma}

\begin{proof}
We will prove by induction on $n$. Consider the base case with only one variable $x_1$.  We can now apply the proof of proposition 6.1.B of \cite{mas1995microeconomic}. We will repeat their proof here:  

$$\int_{\underline{x}_1}^{\Bar{x}_1}f_1(x_1) \cdot u(x_1) \, d x_1 -\int_{\underline{x}_1}^{\Bar{x}_1} g_1(x_1) \cdot u(x_1) \, d x_1 = $$

$$\int_{\underline{x}_1}^{\Bar{x}_1} u(x_1) \cdot (f_1(x_1)-g_1(x_1)) \, d x_1 .$$

    Using integration by parts, this integral is equal to

    $$\Big[u(x_1) \cdot (F_1(x_1)-G_1(x_1))\Big]_{\underline{x}_{1}}^{\Bar{x}_{1}}-\int_{\underline{x}_1}^{\Bar{x}_1} \frac{\partial u(x_1)}{\partial x_1}  \cdot (F_1(x_1)-G_1(x_1)) \, d x_1 =$$

    $$u(\Bar{x}_1) \cdot (1-1) - u(\underline{x}_1) \cdot (0-0) -\int_{\underline{x}_1}^{\Bar{x}_1} \frac{\partial u(x_1)}{\partial x_1}  \cdot (F_1(x_1)-G_1(x_1)) \, d x_1 \geq 0.$$

Note that the first and second term are zero. For the third term, we have $ \frac{\partial u(x_1)}{\partial x_1} \geq 0 $ and we have $F_1(x_1)-G_1(x_1) \leq 0$ for all $x_1$, since $F_1$ first order stochastically dominates $G_1$ .
\newline
\newline
Now consider the case when there are $n \geq 2$ variables and assume that for $n-1$ variables that 

$$\int_{\underline{x}_1}^{\Bar{x}_1} \dots \int_{\underline{x}_{n-1}}^{\Bar{x}_{n-1}}u (x_1,...,x_{n}) \cdot ( \prod_{j=1 }^{n-1}f_j(x_j)- \prod_{j=1 }^{n-1} g_j(x_j) ) \,dx_{n-1} \dots dx_1 \geq 0$$

 for every $x_n$. We should show that 

$$\int_{\underline{x}_1}^{\Bar{x}_1} \dots \int_{\underline{x}_n}^{\Bar{x}_n}u (x_1,...,x_{n}) \cdot (  \prod_{j=1 }^n f_j(x_j)-  \prod_{j=1 }^n g_j(x_j) ) \,dx_{n} \dots dx_1 \geq 0$$

Using integration by parts on the inner integral, this integral is equal to

$$ \int_{\underline{x}_1}^{\Bar{x}_1} \dots\int_{\underline{x}_{n-1}}^{\Bar{x}_{n-1}} \Big[u(x_1,...,x_n)\cdot (F_n(x_n) \cdot \prod_{j=1 }^{n-1} f_j(x_j) - G_n(x_n) \cdot \prod_{j=1 }^{n-1} g_j(x_j))\Big]_{\underline{x}_{n}}^{\Bar{x}_{n}}  \,dx_{n-1} \dots dx_1 $$ 
$$- \int_{\underline{x}_n}^{\Bar{x}_n} \int_{\underline{x}_1}^{\Bar{x}_1} \dots\int_{\underline{x}_{n-1}}^{\Bar{x}_{n-1}} \frac{\partial u (x_1,...,x_{n})}{\partial x_n}  \cdot  (F_n(x_n) \cdot \prod_{j=1}^{n-1}f_j(x_j) - G_n(x_n) \cdot \prod_{j=1}^{n-1} g_j(x_j))   \,dx_{n-1}  \dots dx_1 dx_n $$
As $F(\Bar{x}_n)=G(\Bar{x}_n)=1$ and $F(\underline{x}_n)=G(\underline{x}_n)=0$, this can be rewritten to 
$$\int_{\underline{x}_1}^{\Bar{x}_1} \dots \int_{\underline{x}_{n-1}}^{\Bar{x}_{n-1}}u (x_1,...,\Bar{x}_n) \cdot ( \prod_{j=1}^{n-1} f_j(x_j)- \prod_{j=1}^{n-1} g_j(x_j) ) \,dx_{n-1} \dots dx_1 $$ 
$$- \int_{\underline{x}_n}^{\Bar{x}_n} \int_{\underline{x}_1}^{\Bar{x}_1} \dots\int_{\underline{x}_{n-1}}^{\Bar{x}_{n-1}} \frac{\partial u (x_1,...,x_{n})}{\partial x_n}  \cdot  (F_n(x_n) \cdot\prod_{j=1}^{n-1}f_j(x_j) - G_n(x_n) \cdot \prod_{j=1}^{n-1} g_j(x_j))   \,dx_{n-1}  \dots dx_1 dx_n $$

Note that by the assumption for $n-1$ variables, the first term is positive. Now we will show that 

 $$ \int_{\underline{x}_n}^{\Bar{x}_n} \int_{\underline{x}_1}^{\Bar{x}_1} \dots\int_{\underline{x}_{n-1}}^{\Bar{x}_{n-1}} \frac{\partial u (x_1,...,x_{n})}{\partial x_n}  \cdot  (F_n(x_n) \cdot \prod_{j=1}^{n-1} f_j(x_j) - G_n(x_n) \cdot \prod_{j=1}^{n-1} g_j(x_j))   \,dx_{n-1}  \dots dx_1 dx_n \leq 0$$

Using integration by parts, this integral is equal to

 $$\int_{\underline{x}_n}^{\Bar{x}_n} \int_{\underline{x}_1}^{\Bar{x}_1} \dots \int_{\underline{x}_{n-2}}^{\Bar{x}_{n-2}} \Big[(\frac{\partial u (x_1,...,x_{n})}{\partial x_n } \cdot (F_n(x_n) \cdot F_{n-1}(x_{n-1}) \prod_{j=1}^{n-2}f_j(x_j)$$
 
 $$- G_n(x_n) \cdot  G_{n-1}(x_{n-1}) \cdot \prod_{j=1}^{n-2} g_j(x_j))\Big]_{\underline{x}_{n-1}}^{\Bar{x}_{n-1}}  \,dx_{n-2} \dots dx_1 dx_n $$ 
    
    $$- \int_{\underline{x}_{n-1}}^{\Bar{x}_{n-1}} \int_{\underline{x}_n}^{\Bar{x}_n} \int_{\underline{x}_1}^{\Bar{x}_1} \dots \int_{\underline{x}_{n-2}}^{\Bar{x}_{n-2}} \frac{\partial u (x_1,...,x_{n})}{\partial x_n \partial x_{n-1}}  \cdot $$
    
    $$(F_n(x_n) \cdot F_{n-1}(x_{n-1}) \cdot\prod_{j=1}^{n-2} f_j(x_j) - G_n(x_n) \cdot  G_{n-1}(x_{n-1}) \cdot \prod_{j=1}^{n-2} g_j(x_j))   \,dx_{n-2}  \dots dx_1 dx_n dx_{n-1} =$$

    $$\int_{\underline{x}_n}^{\Bar{x}_n} \int_{\underline{x}_1}^{\Bar{x}_1} \dots \int_{\underline{x}_{n-2}}^{\Bar{x}_{n-2}} \frac{\partial u (x_1,...,\Bar{x}_{n-1},x_{n})}{\partial x_n } \cdot (F_n(x_n) \cdot \prod_{j=1}^{n-2}f_j(x_j)- G_n(x_n) \cdot\prod_{j=1}^{n-2} g_j(x_j) ) \,dx_{n-2} \dots dx_1 dx_n $$ 
    
    $$- \int_{\underline{x}_{n-1}}^{\Bar{x}_{n-1}}  \int_{\underline{x}_n}^{\Bar{x}_n} \int_{\underline{x}_1}^{\Bar{x}_1} \dots\int_{\underline{x}_{n-2}}^{\Bar{x}_{n-2}} \frac{\partial^2u (x_1,...,x_{n})}{\partial x_n \partial x_{n-1}}  \cdot  (F_n(x_n) \cdot F_{n-1}(x_{n-1}) \cdot \prod_{j=1}^{n-2} f_j(x_j)$$
    
    $$- G_n(x_n) \cdot G_{n-1}(x_{n-1})\prod_{j=1}^{n-2} g_j(x_j))   \,dx_{n-2}  \dots dx_1 dx_n dx_{n-1} $$

     We have $ \frac{\partial^2u (x_1,...,x_{n})}{\partial x_n \partial x_{n-1}}= 0 $. Hence, the second term is zero. 
     \newline
     \newline
     By applying integration by parts $n-2$ more times in this fashion, the second term will vanish in each iteration as  $ \frac{\partial^2u (x_1,...,x_{n})}{\partial x_n \partial x_{j}} = 0 $ for all $j \neq n$. Hence, after applying integration by parts $n-2$ times, we are left with

     $$ \int_{\underline{x}_n}^{\Bar{x}_n} \frac{\partial u (\Bar{x}_1,...,\Bar{x}_{n-1},x_{n})}{\partial x_n } \cdot (F_n(x_n)- G_n(x_n) ) \,dx_{n} \leq 0$$

     as $ \frac{\partial u (x_1,...,x_{n})}{\partial x_n} \geq 0 $ and  $F_n(x_n)-G_n(x_n) \leq 0$ for all $x_n$.

     Hence, we have shown that 
     
     $$\int_{\underline{x}_1}^{\Bar{x}_1} \dots \int_{\underline{x}_n}^{\Bar{x}_n}u (x_1,...,x_{n}) \cdot ( \prod_{j=1}^n f_j(x_j)- \prod_{j=1}^n g_j(x_j) ) \,dx_{n} \dots dx_1 \geq 0$$

     which completes the proof.
\end{proof}

\subsection*{Proof of Lemma 1}

\begin{proof}
We prove the part on weakly increasing differences. The part on weakly decreasing differences is similar. As $\frac{\partial^2u_i}{\partial c_i \partial \theta_i} \geq 0$ and $\frac{\partial^2u_i}{\partial c_i \partial c_j} \geq 0 \ \forall j \neq i$ , we have that 

$$\frac{\partial u_i(\theta_i',c_i,c_{-i})}{\partial c_i} \geq \frac{\partial u_i(\theta_i,c_i,c_{-i})}{\partial c_i} \ \text{for every} \ \theta_i' \geq \theta_i, c_i \in C_i \ \text{and} \ c_{-i} \in C_{-i}$$ 

$$\text{and} \ \frac{\partial u_i(\theta_i,c_i,c_{-i}')}{\partial c_i} \geq\frac{\partial u_i(\theta_i,c_i,c_{-i})}{\partial c_i}  \ \text{for every} \ c_{-i}'\geq c_{-i}, \theta_i \in \Theta_i \ \text{and} \ c_i \in C_i.$$

Hence, we have that $\ \forall c_i' \geq c_i$, $\theta_i' \geq \theta_i $ and $c_{-i}' \geq c_{-i}$  

$$ u_i(\theta_i',c_i',c_{-i}')-u_i(\theta_i',c_i,c_{-i}') \geq u_i(\theta_i,c_i',c_{-i}') - u_i(\theta_i,c_i,c_{-i}')$$
and 
$$ u_i(\theta_i,c_i',c_{-i}')-u_i(\theta_i,c_i,c_{-i}') \geq u_i(\theta_i,c_i',c_{-i}) - u_i(\theta_i,c_i,c_{-i}),$$

which implies that

$$ u_i(\theta_i',c_i',c_{-i}')-u_i(\theta_i',c_i,c_{-i}') \geq u_i(\theta_i,c_i',c_{-i}) - u_i(\theta_i,c_i,c_{-i}).$$
  
Next, let $ (\theta_i',\beta_i', f_i') \geq (\theta_i,\beta_i, f_i)$. Then, we have that

$$U_i(\theta_i',c_i',\beta_i',f_i')-U_i(\theta_i',c_i,\beta_i',f_i')=$$

$$  \int_{c_{-i} \in C_{-i}} (\beta_i' \circ f_i')(c_{-i}) \cdot  u_i(\theta_i', c_i',c_{-i}) \, d c_{-i}$$

$$ -  \int_{c_{-i} \in C_{-i}} (\beta_i' \circ f_i')(c_{-i}) \cdot  u_i(\theta_i', c_i,c_{-i}) \, d c_{-i}=$$

$$\int_{c_{-i} \in C_{-i}} (\beta_i' \circ f_i')(c_{-i}) \cdot (u_i(\theta_i', c_i',c_{-i})-u_i(\theta_i', c_i,c_{-i})) \, d c_{-i}.$$

Similarly, we have that

$$U_i(\theta_i,c_i',\beta_i,f_i)-U_i(\theta_i,c_i,\beta_i,f_i)=$$

$$\int_{c_{-i} \in C_{-i}} (\beta_i \circ f_i)(c_{-i}) \cdot (u_i(\theta_i, c_i',c_{-i})-u_i(\theta_i, c_i,c_{-i})) \, d c_{-i}.$$

Hence, we have  weakly increasing differences if

$$\int_{c_{-i} \in C_{-i}}(\beta_i' \circ f_i')(c_{-i}) \cdot (u_i(\theta_i', c_i',c_{-i})-u_i(\theta_i', c_i,c_{-i})) \, d c_{-i} \geq$$ $$\int_{c_{-i} \in C_{-i}}(\beta_i \circ f_i)(c_{-i}) \cdot (u_i(\theta_i, c_i',c_{-i})-u_i(\theta_i, c_i,c_{-i})) \, d c_{-i}  $$

Let $\Delta u_i (\theta_i',c_{-i})=u_i(\theta_i', c_i'c_{-i})-u_i(\theta_i', c_i,c_{-i})$ and $\Delta u_i(\theta_i,c_{-i})=u_i(\theta_i, c_i',c_{-i})-u_i(\theta_i, c_i,c_{-i})$, which are both increasing in $c_j \ \forall j \neq i$ and $\Delta u_i (\theta_i',c_{-i}) \geq \Delta u_i(\theta_i,c_{-i})$ for all $c_{-i} \in C_{-i}$.

Without loss of generality, consider player 1. To show increasing differences we must show that 

$$\int_{\underline{c}_2}^{\Bar{c}_2} \dots \int_{\underline{c}_{n}}^{\Bar{c}_{n}} \Delta u_1 (\theta_1,c_{-1}) \cdot  \prod_{j\neq 1}(\beta_1' \circ f_1')(c_j) \,dc_{n} \dots dc_{2} \geq \int_{\underline{c}_2}^{\Bar{c}_2} \dots \int_{\underline{c}_{n}}^{\Bar{c}_{n}} \Delta u_1 (\theta_1,c_{-1}) \cdot \prod_{j\neq 1}(\beta_1 \circ f_1)(c_j) ) \,dc_{n} \dots dc_{2}$$

As $\frac{\partial^3u_1}{\partial c_1 \partial c_j \partial c_{l}} = 0 $ for all $j \neq l$ and $j,l \neq 1$, we have that $ \frac{\partial^2 \Delta u_1 (\theta_1,c_{-1})}{\partial c_j \partial c_{l}} =\frac{\partial^2 (u_1(\theta_1, c_1',c_{-1})-u_1(\theta_1, c_1,c_{-1}))}{\partial c_j \partial c_{l}} = 0 $. Similarly, as $\frac{\partial^2u_1}{\partial c_1 \partial c_j} \geq 0 $ for all $j \neq 1$, we have that $ \frac{\partial \Delta u_1 (\theta_1,c_{-1})}{\partial c_{j}} =\frac{\partial (u_1(\theta_1, c_1',c_{-1})-u_1(\theta_1, c_1,c_{-1}))}{\partial c_{j}} \geq 0 $ for all $j \neq 1$. Lastly, as $(\beta_1' \circ F_1')\geq (\beta_1 \circ F_1)$, Lemma \ref{mass colel} implies that the above inequality is true.
\end{proof}

\subsection*{Proof of Lemma \ref{monotonicity}}

\begin{proof}
If $c_i=c_i'$, then the statement holds. Next, let $c_i \neq c_i'$ The choice  $c_i$ is optimal for $(\theta_i,\beta_i,f_i)$ and  $c_i'$ is optimal for $(\theta_i',\beta_i',f_i')$, which implies that 

 \begin{equation*}
    U_i(\theta_i,c_i,\beta_i,f_i) - U_i(\theta_i,c_i',\beta_i,f_i) > 0  
 \end{equation*}
 and
  \begin{equation*}
   U_i(\theta_i',c_i,\beta_i',f_i')- U_i(\theta_i',c_i',\beta_i',f_i') < 0 . 
 \end{equation*}
 By contradiction, assume that $c_i' < c_i$. For $U_i$ to have weakly increasing differences, it must then be true that 
 
 $$U_i(\theta_i',c_i,\beta_i',f_i')-U_i(\theta_i',c_i',\beta_i',f_i') \geq U_i(\theta_i,c_i,\beta_i,f_i) -  U_i(\theta_i,c_i',\beta_i,f_i) $$
 
However, our first two inequalities imply that 
 \begin{equation*}
  U_i(\theta_i',c_i,\beta_i',f_i')-U_i(\theta_i',c_i',\beta_i',f_i') <  U_i(\theta_i,c_i,\beta_i,f_i) -  U_i(\theta_i,c_i',\beta_i,f_i), 
 \end{equation*}
 
which is a contradiction. Hence, it must be that $c_i'\geq c_i$.
\end{proof}

\subsection*{Proof of Theorem 1}
\begin{proof}
We will start by proving that $P_i^k(\theta_i)=[l_i^k(\theta_i),u_i^k(\theta_i)]$ for all $i \in N$. We will prove by induction. The statement is true for round $k=0$ as $P_i^0(\theta_i)=[\underline{c}_i,\Bar{c}_i] \ \forall i \in \{1,...,N\}$ and $\theta_i \in \Theta_i$. Now assume that the statement holds for the first $k-1$ rounds. If $P_i^k(\theta_i)$ is a singleton, then the statement is true for round $k$. Hence, assume that $P_i^k(\theta_i)$ is not a singleton. Consider $c_i',c_i'' \in P_i^k(\theta_i)$ with $c_i'' < c_i'$. Now consider some choice $c_i'' < c_i < c_i'$.
\newline
\newline
The choice $c_i''$ is optimal for some choice belief $\beta_i''$ such that $\beta_{ij}''(\theta_{j}) \in P_{j}^{k-1}(\theta_{j})=[l_j^{k-1}(\theta_j),u_j^{k-1}(\theta_j)]$ \newline $\forall \theta_j \in \Theta_j \ \text{and} \ \forall j \in  N \setminus \{i\}$ and parameter belief $f_i'' \in M$. The choice $c_i'$ is optimal for some belief $\beta_i'$ such that $\beta_{ij}'(\theta_{j}) \in P_{j}^{k-1}(\theta_{j})$ $=[l_j^{k-1}(\theta_j),u_j^{k-1}(\theta_j)]  \ \forall \theta_j \in \Theta_j \ \text{and} \ \forall j \in  N \setminus \{i\}$ and parameter belief $f_i' \in M$. Now consider a belief $\beta_i$ such that $\beta_{ij}(\theta_j)=(1-\lambda) \cdot \beta_{ij}''(\theta_j) + \lambda \cdot \beta_{ij}'(\theta_{j}) \ \forall \theta_j \in \Theta_j \ \text{and} \ \forall j \in  N \setminus \{i\} $, where $\lambda \in [0,1]$. Hence $\beta_{ij}(\theta_j) \in [l_j^{k-1}(\theta_j),u_j^{k-1}(\theta_j)]=P_{j}^{k-1}(\theta_{j})  \ \forall \theta_j \in \Theta_j \ \text{and} \ \forall j \in  N \setminus \{i\}$.
\newline
\newline
 Assumption \ref{continuity} implies that there exists some value of $\lambda \in [0,1]$ for the choice belief $\beta_i= (1-\lambda) \cdot \beta_i'' + \lambda \cdot \beta_i'$ and parameter belief $f_i= (1-\lambda) \cdot f_i'' + \lambda \cdot f_i'$  such that $c_i$ is optimal for $(\theta_i, \beta_i, f_i)$ , which implies that $c_i \in P_i^k(\theta_i)$. This implies that all choices in between $l_i^k(\theta_i)$ and $u_i^k(\theta_i)$ are also in $P_i^k(\theta_i)$. As a result, we have that $P_i^k(\theta_i)=[l_i^k(\theta_i),u_i^k(\theta_i)]$.  Hence, the statement also holds for round $k$.
\newline
\newline
Next, we will prove that if $\theta_i < \theta_i'$, then $l_i^k(\theta_i) \leq l_i^k(\theta_i')$ and $u_i^k(\theta_i) \leq u_i^k(\theta_i')$ for every $k \in \mathbb{N}$. We will prove the first statement by induction on $k$. Note that the base case $k=0$ is true. Now assume that our statement holds for the first $k-1$ rounds of the iterative procedure.
\newline
\newline
We will first consider the lower bound. By contradiction, let $l_i^k(\theta_i) > l_i^k(\theta_i')$ after round $k$. This means that the choice $l_i^k(\theta_i')$ will be eliminated for player $i$ with parameter $\theta_i$ after round $k$. Hence, for parameter $\theta_i$, for any choice belief $\beta_i$, where $\beta_{ij}(\theta_{j}) \in P_{j}^{k-1}(\theta_{j}) \ \forall \theta_j \in \Theta_j \ \text{and} \ \forall j \in  N \setminus \{i\}$ and any parameter belief $f_i \in M$, there will be some choice in $P_i^{k-1}(\theta_i)$ that yields a greater utility than choosing $l_i^k(\theta_i')$. We denote this by 
\begin{equation}\label{star}
   U_i(\theta_i,l_i^k(\theta_i'),\beta_i, f_i) <  U_i(\theta_i,c_i,\beta_i,f_i) \ \text{for some} \ c_i \in P_i^{k-1}(\theta_i) . 
\end{equation}

Furthermore, we have that the choice $l_i^k(\theta_i')$ is optimal for player $i$ with parameter $\theta_i'$. Hence, there exists some choice belief $\beta_i'$, where $\beta_{ij}'(\theta_{j}) \in P_{j}^{k-1}(\theta_{j}) \ \forall \theta_j \in \Theta_j \ \text{and}  \ \forall j \in  N \setminus \{i\}$ and parameter belief $f_i' \in M$, such that  

\begin{equation}\label{star2}
   U_i(\theta_i',l_i^k(\theta_i'),\beta_i',f_i') \geq U_i(\theta_i',c_i,\beta_i',f_i') \ \forall \ c_i \in P_i^{k-1}(\theta_i'). 
\end{equation}

For this choice belief $\beta_i'$ and parameter belief $f_i'$, consider a choice $c_i'$ such that $U_i(\theta_i,c_i',\beta_i',f_i') \geq U_i(\theta_i,c_i,\beta_i',f_i') \ \forall c_i \in P_i^{k-1}(\theta_i)$. Furthermore, (\ref{star}) implies that $U_i(\theta_i,l_i^k(\theta_i'),\beta_i',f_i') <  U_i(\theta_i,c_i',\beta_i',f_i')$. It must be that $c_i' \in P_i^{k}(\theta_i)$, which implies that $c_i' \geq l_i^k(\theta_i)$. Furthermore, we have that $c_i' \in P_i^{k}(\theta_i) \subseteq P_i^{k-1}(\theta_i)$, which implies that $c_i' \leq u_i^{k-1}(\theta_i)$. As a result, we have that

 \begin{equation}\label{1}
 U_i(\theta_i,l_i^k(\theta_i'),\beta_i',f_i') <  U_i(\theta_i,c_i',\beta_i',f_i') \ \text{for some} \ \ c_i' \in [l_i^k(\theta_i),u_i^{k-1}(\theta_i)] \subseteq P_i^{k-1}(\theta_i').
 \end{equation}
 We have that $ [l_i^k(\theta_i),u_i^{k-1}(\theta_i)] \subseteq P_i^{k-1}(\theta_i')$ as $l_i^k(\theta_i) > l_i^k(\theta_i') \geq l_i^{k-1}(\theta_i')$ and $u_i^k(\theta_i) \leq u_i^{k-1}(\theta_i) \leq u_i^{k-1}(\theta_i')$. Together with (\ref{star2}) this implies that
  \begin{equation}\label{2}
U_i(\theta_i',l_i^k(\theta_i'),\beta_i',f_i') \geq U_i(\theta_i',c_i',\beta_i',f_i') 
 \end{equation}
 
We can rewrite (\ref{1}) and (\ref{2}) to
 
  \begin{equation}\label{3}
 U_i(\theta_i,c_i',\beta_i',f_i') - U_i(\theta_i,l_i^k(\theta_i'),\beta_i',f_i') > 0 
 \end{equation}
 and 
  \begin{equation}\label{4}
U_i(\theta_i',c_i',\beta_i',f_i') - U_i(\theta_i',l_i^k(\theta_i'),\beta_i',f_i')   \leq 0. 
 \end{equation}
 Finally, we have that (\ref{3}) and (\ref{4}) imply
  \begin{equation}\label{5}
 U_i(\theta_i',c_i',\beta_i',f_i') - U_i(\theta_i',l_i^k(\theta_i'),\beta_i',f_i') < U_i(\theta_i,c_i',\beta_i',f_i') - U_i(\theta_i,l_i^k(\theta_i'),\beta_i',f_i')  
 \end{equation}
 
 As $c_i' > l_i^k(\theta_i')$ and $\theta_i' > \theta_i $, this is a contradiction to Assumption 1. Hence, it must be that $l_i^k(\theta_i) \leq l_i^k(\theta_i')$ after round $k$. In a similar way, it can be shown that $u_i^k(\theta_i') \geq u_i^k(\theta_i)$.
 \newline
 \newline
 Next, we will show that $P_i(\theta_i)=[ l_i(\theta_i),u_i(\theta_i)]$ and $l_i(\theta_i) \leq l_i(\theta_i') \ \text{and} \ u_i(\theta_i) \leq u_i(\theta_i')$.
 \newline
 \newline
 As $P_i^k(\theta_i)=[l_i^k(\theta_i),u_i^k(\theta_i)]$ for all $i \in N$, we have  $P_i(\theta_i)= \underset{k \geq 1}{\bigcap} P_i^k(\theta_i)= \underset{k \geq 1}{\bigcap} [l_i^k(\theta_i),u_i^k(\theta_i)]$. Similarly we have $P_i(\theta_i')= \underset{k \geq 1}{\bigcap} P_i^k(\theta_i')= \underset{k \geq 1}{\bigcap} [l_i^k(\theta_i'),u_i^k(\theta_i')]$. Furthermore, we have  $l_i^k(\theta_i) \leq l_i^k(\theta_i')$ and $u_i^k(\theta_i) \leq u_i^k(\theta_i') \ \forall k \geq 1$. Hence, $P_i(\theta_i)=[ l_i(\theta_i),u_i(\theta_i)] \ \text{and} \ P_i(\theta_i')=[ l_i(\theta_i'),u_i(\theta_i')]$, where $l_i(\theta_i) \leq l_i(\theta_i') \ \text{and} \ u_i(\theta_i) \leq u_i(\theta_i')$.
 \newline
 \newline
 Consider player $i$ with parameter $\theta_i$, choice belief $\beta_i'$ and parameter belief $f_i'$ as described in Theorem 1. We will now show that $(\beta_i'\circ f_i') \geq (\beta_i \circ f_i) \ \forall \beta_i$ such that $\beta_{ij}(\theta_j) \in P_j^{k-1}(\theta_j) \ \forall \theta_j \in \Theta_j \ \text{and} \ j \neq i $ and $\forall f_i \in M$. Hence, we will show that 

    $$\int_{c_j'}^{\Bar{c}_{j}} (\beta_{ij}' \circ f_{ij}')(c_{j}) \, d c_{j} \geq \int_{c_{j}'}^{\Bar{c}_{j}} (\beta_{ij} \circ f_{ij})(c_j)  \, d c_{j}  \ \forall c_{j}'  \in C_{j} \ \text{and}\ j\neq i$$

    which can be rewritten to

    $$\int_{\theta_j:\beta_{ij}'(\theta_j)\geq c_{j}'}  f_{ij}'(\theta_j) \, d \theta_j \geq \int_{\theta_j:\beta_{ij}(\theta_j)\geq c_{j}'}  f_{ij}(\theta_j) \, d \theta_j \ \forall c_{j}'  \in C_{j} \ \text{and}\ j\neq i.$$

 We have that $\beta_{ij}'(\theta_j)=  \max P_j^{k-1}(\theta_j) \ \forall \theta_j \in \Theta_j$. As $ \max P_j^{k-1}(\theta_j)$ is weakly increasing in $\theta_j$, we have that $\beta_{ij}'(\theta_j)$ is weakly increasing in $\theta_j$. Now consider the parameter value $\theta_j'$ such that $\theta_j'= \min \{ \theta_j \in \Theta_j | \beta_{ij}'(\theta_j)= c_j'\}$. Hence, we have that 

    $$\int_{\theta_j:\beta_{ij}'(\theta_j)\geq c_{j}'}  f_{ij}'(\theta_j) \, d \theta_j = \int_{\theta_j'}^{\Bar{\theta_j}}  f_{ij}'(\theta_j) \, d \theta_j$$

    Consider any $\theta_j'' < \theta_j'$. Then $\beta_{ij}(\theta_j'')\leq \beta_{ij}'(\theta_j'') < c_j'$. As a result, 

$$\int_{\theta_j:\beta_{ij}(\theta_j)\geq c_{j}'}  f_{ij}(\theta_j) \, d \theta_j \leq \int_{\theta_j'}^{\Bar{\theta_j}}  f_{ij}(\theta_j) \, d \theta_j \leq \int_{\theta_j'}^{\Bar{\theta_j}}  f_{ij}'(\theta_j) \, d \theta_j = \int_{\theta_j:\beta_{ij}'(\theta_j)\geq c_{j}'}  f_{ij}'(\theta_j) \, d \theta_j .$$

The second inequality holds as $f_{ij}'$ first-order stochasically dominates $f_{ij}$. As this is true for any opponent $j$, this implies that $(\beta_i'\circ f_i') \geq (\beta_i \circ f_i)$. Now let $u_i^k(\theta_i)$ be optimal for $(\theta_i,\beta_i',f_i')$. Note that for any other $\beta_i$ such that $\beta_{ij}(\theta_j) \in P_j^{k-1}(\theta_j) \ \forall \theta_j \in \Theta_j \ \text{and} \ j \neq i $ and $\forall f_i \in M$ we have that $(\theta_i,\beta_i',f_i') \geq (\theta_i,\beta_i,f_i)$. Then, by Lemma \ref{monotonicity}, for the optimal choice $c_i$ that corresponds to $(\theta_i,\beta_i,f_i)$ we have that $c_i \leq u_i^k(\theta_i)$. The proof for the lower bound is similar.   
    
 \end{proof}

 \newpage

\bibliographystyle{apacite}
\bibliography{references}

\end{document}